\documentclass[a4paper,oneside,fleqn,12pt]{article}
\usepackage{amsmath,amssymb,amsfonts,amsthm,type1cm,bm,color,colortbl, mathrsfs}
\usepackage{graphicx,psfrag,epsf,booktabs, geometry}
\usepackage{natbib}
\usepackage{rotating}
\usepackage{authblk}%
\RequirePackage[colorlinks,citecolor=blue,urlcolor=blue,hyperfootnotes=false]{hyperref}
\usepackage{accents}
\usepackage{yfonts}
\usepackage{pdflscape}
\usepackage{setspace}
\usepackage{longtable}
\usepackage{bigstrut}
\usepackage{mathtools}
\usepackage{enumitem}
\mathtoolsset{showonlyrefs=true}
\usepackage[dvipsnames]{xcolor}
\usepackage{comment}
\usepackage{float}
\usepackage{multibib}
\newcites{suppl}{References for Supplementary Material}

\DeclareMathOperator\E{\mathbb{E}}

\DeclareMathOperator\Pro{\mathbb{P}}

\def\cH{\mathcal{H}}

\def\cL{\mathcal{L}}

\def\cN{\mathcal{N}}

\def\cT{\mathcal{T}}

\def\bbR{\mathbb{R}}

\def\bbN{\mathbb{N}}

\newlength{\dhatheight}

\newtheorem{thm}{Theorem}%[section]
\newtheorem{lem}{Lemma}%[section]
\newtheorem{con}{Condition}%[section]
\newtheorem{rem}{Remark}%[section]
\newtheorem{defi}{Definition}%[section]
\newtheorem{prop}{Proposition}%[section]
\makeatletter
\def\section{\@startsection {section}{1}{\z@}{-3.5ex plus -1ex minus-.2ex}{2.3ex plus .2ex}{\large\bf}}
\makeatother

\makeatletter
\def\subsection{\@startsection {subsection}{1}{\z@}{-3.5ex plus -1ex minus-.2ex}{2.3ex plus .2ex}{\normalsize\bf}}
\makeatother

\title{\textbf{A Likelihood-Ratio Test for Verifying Weak Stochastic Transitivity}}

\author[$*$]{\textsc{Masaki Toyoda}%\thanks{Yoshimasa Uematsu is Associate Professor, Department of Social Data Science, Hitotsubashi University, 2-1 Naka, Kunitachi, Tokyo 186-8601, Japan (E-mail: yoshimasa.uematsu@r.hit-u.ac.jp).}
}

\affil[*]{\textit{Department of Economics, Hitotsubashi University}}

\begin{document}

\renewcommand{\theequation}{\thesection.\arabic{equation}}
\makeatletter
\@addtoreset{equation}{section}
\makeatother

\maketitle

\begin{abstract}
This paper proposes and studies a likelihood-ratio test of the null hypothesis that weak stochastic transitivity (WST) does not hold. This is the reverse of the formulation commonly used in the literature, where WST is set as the null hypothesis. We show that, even when the number of items grows with the number of comparisons per pair, the uniform size converges to the nominal level, with the critical value determined by a chi-bar-square distribution. We further establish that the Type II error converges uniformly to zero under a sufficient signal-strength condition. Simulation studies demonstrate good finite-sample Type I error control and show that power increases with the sample size and signal strength.
\end{abstract}

\section{Introduction}\label{sec:intro}
In this paper, we study statistical inference for pairwise comparison data, which represent the outcomes of comparisons between pairs of items. Such data are widely used to estimate a ranking for a set of items, and arise in many areas such as psychology \citep{thurstone1927method}, chess \citep{zermelo1929die, elo1967proposed}, sports \citep{massey1997statistical}, dueling bandits \citep{yue2009interactively, jamieson2015sparse}, and reinforcement learning from human feedback (RLHF) \citep{rafailov2023direct, chiang2024chatbot}.

To analyze pairwise comparison data, researchers often assume transitivity. Several notions of stochastic transitivity have been proposed: see, for example, \cite{oliveira2018stochastic}. In this paper, we focus on one of the least restrictive notions, weak stochastic transitivity (WST). WST requires that, for any items $i,j,k$, if $i$ beats $j$ with probability at least $1/2$ and $j$ beats $k$ with probability at least $1/2$, then $i$ must also beat $k$ with probability at least $1/2$. Many standard models impose assumptions that imply WST. For example, the Bradley-Terry model \citep{bradley1952rank} and the Thurstone model \citep{thurstone1927method}, two of the most widely used models for pairwise comparisons, assume that each item has a one-dimensional score, and the ordering of these scores implies WST. Therefore, if the transitivity assumption is violated, then the results of analyses based on such models can be difficult to interpret.

Nevertheless, intransitive patterns have been empirically reported in a variety of settings, including psychology \citep{tversky1969intransitivity} and e-sports \citep{chen2016modeling, duan2017generalized, makhijani2019parametric}. Hypothesis tests have been developed to detect violations of transitivity, typically taking transitivity as the null hypothesis. For example, \cite{iverson1985statistical} developed an asymptotic likelihood-ratio test. In empirical studies, \cite{shafir1994intransitivity} and \cite{waite2001intransitive} used binomial tests for individual pairs, sometimes with multiplicity adjustments. Such tests are useful for detecting violations of WST. However, failure to reject WST does not provide affirmative evidence that WST holds; it only indicates that the available evidence against WST is insufficient.

To obtain affirmative evidence for WST, we need a test in the opposite direction: the null hypothesis is that WST does not hold. Rejection can support the use of procedures whose validity requires WST alone, such as \cite{falahatgar2018limits}.

In related work, \cite{haddenhorst2021testing} proposed a sequential test based on multiple binomial tests. However, their theory assumes that every pairwise comparison probability is different from $1/2$. Therefore, their test has no theoretical guarantee at the boundary cases of WST. We revisit this point in Section \ref{sec:simulations}.

In this paper, we propose and investigate a likelihood-ratio test for verifying WST. The main theoretical contributions are as follows. Theorem \ref{thm:m3reduction} establishes an exact finite-sample reduction of the worst-case Type I error to the three-item problem. Theorem \ref{thm:type1} characterizes the asymptotic size of the test for any number of items. Theorem \ref{thm:power} shows that the Type II error converges uniformly to zero under an appropriate signal-strength condition.

This paper is organized as follows. Section \ref{sec:problem} introduces the likelihood-ratio test. Section \ref{sec:theory} states the theoretical results for the proposed test. Section \ref{sec:simulations} investigates the finite-sample behavior of the test through simulations. Section \ref{sec:conclusion} concludes.

\section{Problem Statement}\label{sec:problem}
\subsection{Model}
Consider the set of $m\in\bbN$ items $[m]:=\{1,\dots,m\}$. For each pair of items $i,j\in[m]$ with $i<j$, we observe $n\in\bbN$ comparisons modeled as $X_{1,ij},X_{2,ij},\dots,X_{n,ij}\sim\text{Ber}(p_{ij})$, where $p_{ij}=1-p_{ji}$ is the probability that $i$ beats $j$. Suppose that all comparisons are mutually independent.

We focus on weak stochastic transitivity (WST) \citep{tversky1969intransitivity}, defined as follows.

\begin{defi}
    WST holds if for all distinct $i,j,k$,
    \begin{align}
        p_{ij}\geq1/2~~\text{and}~~p_{jk}\geq1/2~~\Longrightarrow~~p_{ik}\geq1/2.
    \end{align}
\end{defi}
If we write $i\succeq j$ when $p_{ij}\geq1/2$, WST states that $\succeq$ is transitive and defines a ranking with possible ties. See, for example, \citet{regenwetter2010testing} and \citet{falahatgar2018limits}. 

Our goal is to develop a test for
\begin{align}\label{eq:hypothesis}
    H_0:\text{WST does not hold}~~\text{vs.}~~H_1:\text{WST holds}.
\end{align}
The null hypothesis $H_0$ is equivalent to the existence of distinct $i,j,k$ such that $p_{ij}\geq1/2$, $p_{jk}\geq1/2$ and $p_{ik}<1/2$.

Let $p=(p_{ij})_{1\leq i<j\leq m}$ and let $\Theta_m=[0,1]^{m(m-1)/2}$ be the parameter space. Define the null parameter space by $\cH_{0m}=\{p\in\Theta_m:\text{WST does not hold}\}$. For example, when $m=3$, we can write
\begin{align}
    \cH_{03}=\left\{p\in\Theta_3:
    \begin{aligned}
        &\left(p_{12}\geq\frac{1}{2},p_{23}\geq\frac{1}{2},p_{13}\leq\frac{1}{2}, (p_{12},p_{23},p_{13})\neq\left(\frac{1}{2},\frac{1}{2},\frac{1}{2}\right)\right)\\
        &\text{or}
        \left(p_{12}\leq\frac{1}{2},p_{23}\leq\frac{1}{2},p_{13}\geq\frac{1}{2}, (p_{12},p_{23},p_{13})\neq\left(\frac{1}{2},\frac{1}{2},\frac{1}{2}\right)\right)
    \end{aligned}
    \right\}
\end{align}
Figure \ref{fig:nullregion} exhibits $\cH_{03}$. The set $\cH_{03}$ includes all its boundary points except $(1/2,1/2,1/2)$. For example, $p_{12}>1/2$, $p_{23}>1/2$ and $p_{13}=1/2$ is included in $\cH_{03}$ because it holds that $p_{31}\geq1/2$, $p_{12}\geq1/2$ and $p_{32}<1/2$.

\begin{figure}[htpb]
    \centering
    \includegraphics[width=0.8\linewidth]{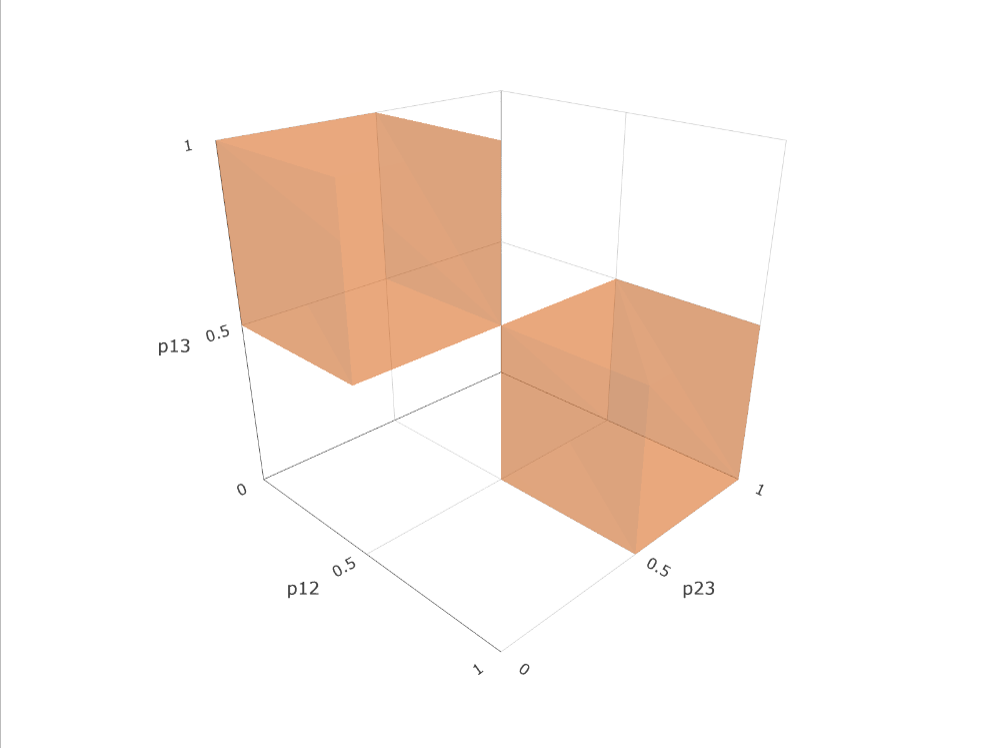}
    \caption{Colored region except $(1/2,1/2,1/2)$ is $\cH_{03}$.}
    \label{fig:nullregion}
\end{figure}

\subsection{Likelihood-ratio test}
Our test statistic for \eqref{eq:hypothesis} is the likelihood-ratio:
\begin{align}
        T_{n,m}=2\log\frac{\sup_{p\in\Theta_m}L_{n,m}(p)}{\sup_{p\in\cH_{0m}}L_{n,m}(p)},
\end{align}
where
\begin{align}
        L_{n,m}(p)=\prod_{1\leq i<j\leq m}p_{ij}^{n\hat{p}_{n,ij}}(1-p_{ij})^{n-n\hat{p}_{n,ij}},~~
        \hat{p}_{n,ij}=\frac{1}{n}\sum_{t=1}^{n}X_{t,ij}.
\end{align}
We can directly compute $T_{n,m}$ as follows. For any $i<j<k$, let $\hat{r}_{ijk}^{(1)}=\hat{p}_{n,ij}$, $\hat{r}_{ijk}^{(2)}=\hat{p}_{n,jk}$ and $\hat{r}_{ijk}^{(3)}=1-\hat{p}_{n,ik}$. Define $D_n^-(x)=2n\text{KL}\left(x\;\middle\|\;1/2\right)\mathbf{1}\{x<1/2\}$ and $D_n^+(x)=2n\text{KL}\left(x\;\middle\|\;1/2\right)\mathbf{1}\{x>1/2\}$, where
\begin{align}
    \text{KL}\left(x\;\middle\|\;1/2\right)=\begin{cases}
            x\log(2x)+(1-x)\log(2-2x)~&x\in(0,1)\\
            \log2~&x\in\{0,1\}.
        \end{cases}
\end{align}
Then, we have
\begin{align}
    T_{n,m}=\min_{i<j<k}\min\left(\sum_{l=1}^{3}D_n^-(\hat{r}_{ijk}^{(l)}),\sum_{l=1}^{3}D_n^+(\hat{r}_{ijk}^{(l)})\right).
\end{align}
For a critical value $c>0$, our likelihood-ratio test rejects $H_0$ if $T_{n,m}>c$.

\section{Theory}\label{sec:theory}
Throughout this paper, we write $\Pro_{p}$ to emphasize that the probability is evaluated under the true parameter $p\in\Theta_m$. In this section, we investigate the asymptotic behavior of $\Pro_{p}(T_{n,m}>c)$.

\subsection{Type I error control}\label{sec:theorytype1}
We first establish an exact finite-sample reduction for the worst-case Type I error.

\begin{thm}\label{thm:m3reduction}
    For any $n\in\bbN$, any $m\geq3$ and any $c>0$, it holds that
    \begin{align}\label{eq:m3reduction}
        \sup_{p\in\cH_{0m}}\Pro_p(T_{n,m}>c)=\sup_{p\in\cH_{03}}\Pro_p(T_{n,3}>c).
    \end{align}
\end{thm}
The proof of Theorem \ref{thm:m3reduction} is given in Appendix \ref{proofthm:m3reduction}. Although $T_{n,m}$ is the minimum over $\binom{m}{3}$ triples, Theorem \ref{thm:m3reduction} states that the size of the test is exactly the same as that for $m=3$, for any $n\in\bbN$ and any $c>0$.

The following theorem is our main result for the asymptotic Type I error guarantee.

\begin{thm}\label{thm:type1}
    Let $\{m_n\}_n$ be any sequence of integers satisfying $m_n\geq3$. Then, for any $c>0$, it holds that
        \begin{align}
            \lim_{n\to\infty}\sup_{p\in\cH_{0m_n}}\Pro_{p}(T_{n,m_n}>c)=\frac{1}{2}\Pro(\chi_1^2>c)+\frac{1}{4}\Pro(\chi_2^2>c).
        \end{align}
\end{thm}
The proof of Theorem \ref{thm:type1} is given in Appendix \ref{proofthm:type1}. According to this theorem, if $c=c_\alpha$ is chosen so that $(1/2)\Pro(\chi_1^2>c_\alpha)+(1/4)\Pro(\chi_2^2>c_\alpha)=\alpha$, then the Type I error of our test $\mathbf{1}\{T_{n,m}>c_\alpha\}$ is asymptotically controlled by $\alpha$. For example, $c_\alpha\approx2.95,4.23,7.29$ for $\alpha=0.1,0.05,0.01$, respectively. Below, we provide several remarks on Theorem \ref{thm:type1}.

\begin{rem}\normalfont\textbf{Uniform size control.}
    The supremum in Theorem \ref{thm:type1} is taken over the null parameter space $\cH_{0m_n}$, which may change with $n$. Thus, the theorem provides uniform rather than pointwise asymptotic control. For example, when $m_n=m$ is fixed, a pointwise guarantee
    \begin{align}
        \limsup_{n\to\infty}\Pro_p(T_{n,m}>c_\alpha)\leq\alpha~~\text{for every fixed $p\in\cH_{0m}$}
    \end{align}
    does not rule out a sequence $q_n\in\cH_{0m}$ such that $\limsup_{n\to\infty}\Pro_{q_n}(T_{n,m}>c_\alpha)>\alpha$. In contrast, Theorem \ref{thm:type1} implies that for every sequence $p_n\in\cH_{0m_n}$, the limit superior of the false rejection probability is bounded by the same value.

    Moreover, Theorem \ref{thm:type1} gives the exact limiting size, rather than only an upper bound on its limit superior. Therefore, the asymptotic bound is sharp. That is, there exists a sequence $p_n\in\cH_{0m_n}$ such that $\Pro_{p_n}(T_{n,m_n}>c_\alpha)\to\alpha$. One such sequence is given by $p_{12}=p_{23}=1/2$, $p_{13}=0$, and $p_{ij}=1$ for all $n$ and $1\leq i<j\leq m_n$ with $j\geq4$, as used in the proof of Theorem \ref{thm:m3reduction}.
\end{rem}

\begin{rem}\normalfont\textbf{Arbitrary number of items.}
    Theorem \ref{thm:type1} holds for any sequence $\{m_n\}_n$ satisfying $m_n\geq3$, without any restriction on its growth rate relative to $n$. Thus, the number of items may be fixed or may diverge arbitrarily fast. This dimension-free result follows from the exact finite-sample reduction established in Theorem \ref{thm:m3reduction}.
\end{rem}

\begin{rem}\normalfont\textbf{Chi-bar-square distribution.}
    The asymptotic size in Theorem \ref{thm:type1} is the upper-tail probability of a chi-bar-square distribution. This form can be understood from the local geometry of the null parameter set. General likelihood theory characterizes the limiting likelihood-ratio statistic by the squared distance of a limiting Gaussian vector from the tangent cone of the null parameter set; see \cite{chernoff1954distribution, self1987asymptotic, drton2009likelihood, brazzale2024likelihood}.
    
    According to Theorem \ref{thm:m3reduction}, it is sufficient to consider the case $m_n=3$. As shown in the proof of Theorem \ref{thm:type1}, an asymptotically least-favorable null configuration can be taken as $(p_{12},p_{23},p_{13})=(1/2,1/2,0)$. At this configuration, $\hat{p}_{n,13}=0$ a.s., so the observation for the pair $(1,3)$ does not contribute to the limiting likelihood-ratio. After standardizing the first two coordinates, the relevant tangent cone is $\bbR_+^2=[0,\infty)\times[0,\infty)$. Also, independence across pairwise comparisons implies that the corresponding limiting Gaussian vector $(Z_1,Z_2)$ follows $N(0,I_2)$. Thus, its squared distance from $\bbR_+^2$ is $Z_1^2\mathbf{1}\{Z_1<0\}+Z_2^2\mathbf{1}\{Z_2<0\}$, which has the chi-bar-square distribution $(1/4)\chi_0^2+(1/2)\chi_1^2+(1/4)\chi_2^2$, where $\chi_0^2$ denotes a point mass at zero. Since $c>0$, $\chi_0^2$ does not contribute to the upper-tail probability, yielding the expression in Theorem \ref{thm:type1}.
\end{rem}

\subsection{Type II error control}\label{sec:theorytype2}
In this section, we establish a uniform Type II error guarantee in Theorem \ref{thm:power}.

\begin{defi}\label{def:signalspace}
        For any $m\geq3$ and $\Delta>0$, let $\cH_{1m}(\Delta)=\{p\in\Theta_{m}:\inf_{q\in\cH_{0m}}\|p-q\|_\infty\geq \Delta\}$.
\end{defi}

\begin{con}\label{con:signal}
        Let $\{\Delta_n\}_n$ be a sequence in $(0,1/2]$ and $\{m_n\}_n$ be a sequence of integers satisfying $m_n\geq3$. Assume that $n\Delta_n^2/\log m_n\to\infty$ as $n\to\infty$.
\end{con}
In words, the set $\cH_{1m}(\Delta)$ in Definition \ref{def:signalspace} consists of parameters whose $\ell_\infty$-distance from $\cH_{0m}$ is at least $\Delta$. Condition \ref{con:signal} requires the sequence of signal strengths $\{\Delta_n\}_n$ to dominate $\sqrt{\log m_n/n}$.

\begin{thm}\label{thm:power}
    If Condition \ref{con:signal} holds, then for any $c>0$, it holds that
        \begin{align*}
            \lim_{n\to\infty}\sup_{p\in\cH_{1m_n}(\Delta_n)}\Pro_{p}(T_{n,m_n}\leq c)=0.
        \end{align*}
\end{thm}
The proof of Theorem \ref{thm:power} is given in Appendix \ref{proofthm:power}. Theorem \ref{thm:power} says that uniformly over $\cH_{1m_n}(\Delta_n)$, the Type II error converges to zero. However, Condition \ref{con:signal} is only a sufficient condition. It remains open whether weaker signal strengths lead to the same Type II error control.

\section{Simulations}\label{sec:simulations}
We investigate the finite-sample Type I error and power of our likelihood-ratio test. For each pair $(i,j)$ with $1\leq i<j\leq m$, we generate $X_{1,ij},\dots,X_{n,ij}\sim\text{Ber}(p_{ij})$ independently. We consider $m\in\{3,10,15\}$ and $n\in\{100,300,500,700,900\}$. We set the nominal level to $\alpha=0.05$, which yields the critical value of our test as $c_\alpha\approx4.23$. All reported results are based on $2000$ replications. Sections \ref{sec:simulationtype1} and \ref{sec:simulationpower} investigate the Type I error and power, respectively.

We compare our test with Algorithm 2 of \cite{haddenhorst2021testing}. To the best of our knowledge, their framework is the only existing theoretically guaranteed framework that is applicable to the test for $H_0:\text{WST does not hold}$. However, the theoretical guarantee of their procedure requires a pre-specified $h>0$ such that $|p_{ij}-1/2|>h$ for every pair $(i,j)$; we set $h=0.04$ and revisit this point in the subsequent sections. Further implementation details are provided in Appendix \ref{sec:haddenhorst}.

\subsection{Type I error}\label{sec:simulationtype1}
To investigate the Type I error, we consider two configurations: a least-favorable null case and an interior null case.

For the least-favorable null, we set $p_{12}=p_{23}=1/2$, $p_{13}=0$ and $p_{ij}=1$ for all other $i<j$. Then, the probability of rejection of our likelihood-ratio test converges to $\alpha$ under this configuration. On the other hand, the procedure of \cite{haddenhorst2021testing} has no theoretical guarantee in this configuration because $p_{12}=p_{23}=1/2$, whereas their procedure requires $|p_{ij}-1/2|>h>0$ for every pair $i<j$. The result is shown in the upper row of Figure \ref{fig:type1}. The Type I error of our test is close to the nominal level $\alpha=0.05$ for all $m$ and $n$. In contrast, the Type I error of the test of \cite{haddenhorst2021testing} increases with $n$ and substantially exceeds $\alpha$.

For the interior null, we set $p_{13}=0.4$ and $p_{ij}=0.6$ for all other $i<j$. This configuration lies in the interior of the null parameter space and satisfies $|p_{ij}-1/2|=0.1>0.04=h$ for every pair. Thus, the Type I error of the test of \cite{haddenhorst2021testing} is theoretically guaranteed. The result is shown in the lower row of Figure \ref{fig:type1}. For both tests and all $m$ and $n$, the null hypothesis is not rejected in any of $2000$ replications.

\begin{figure}[htpb]
    \centering
    \includegraphics[width=0.9\linewidth]{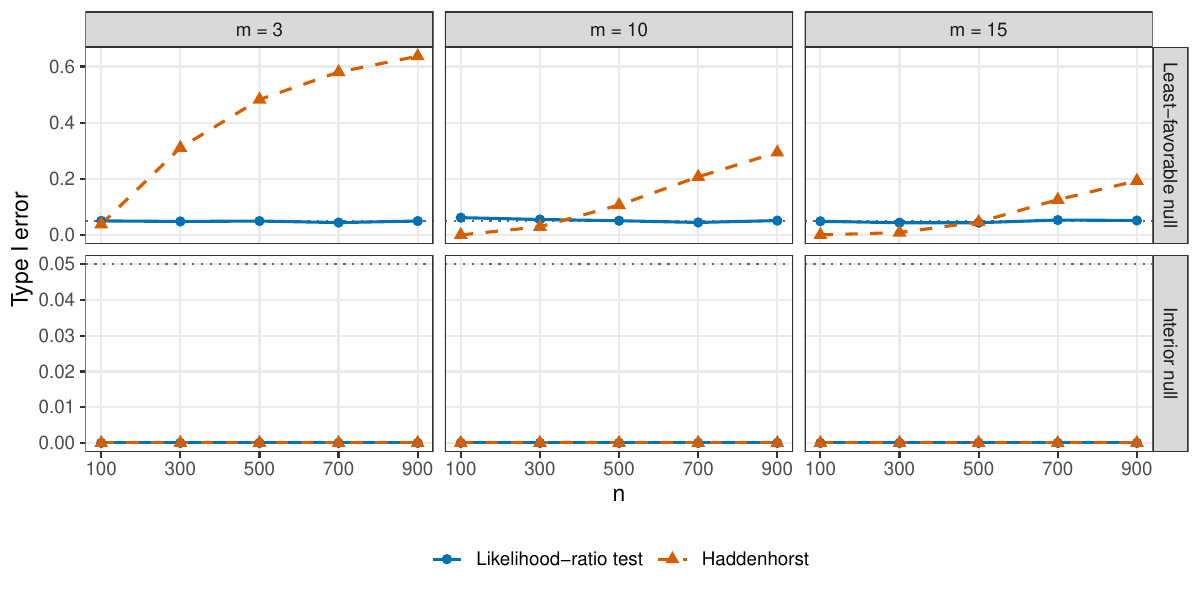}
    \caption{A plot of the Type I error. The blue solid line and the orange dashed line correspond to our likelihood-ratio test and the test of Haddenhorst et al. (2021), respectively. Upper row: least-favorable null. Lower row: interior null. The columns correspond to $m=3,10,15$ from left to right.}
    \label{fig:type1}
\end{figure}

\subsection{Power}\label{sec:simulationpower}
To investigate the power, for every pair $i<j$, we set $p_{ij}=1/2+\Delta$ with $\Delta\in\{0.05,0.10,0.15\}$. These configurations satisfy WST with the order $1\succeq2\cdots\succeq m$. Moreover, all these configurations satisfy $|p_{ij}-1/2|>h=0.04$ required by the test of \cite{haddenhorst2021testing}. We report the proportion of rejections across $2000$ replications as the power. The result is shown in Figure \ref{fig:power}.

For both tests, power generally increases with $n$ and $\Delta$. However, the relative performance of these tests depends on $\Delta$. For $\Delta=0.05$, the test of \cite{haddenhorst2021testing} generally has higher power. In contrast, for $\Delta=0.15$, our likelihood-ratio test attains high power more quickly, particularly when $m=15$.

\begin{figure}[htpb]
    \centering
    \includegraphics[width=0.9\linewidth]{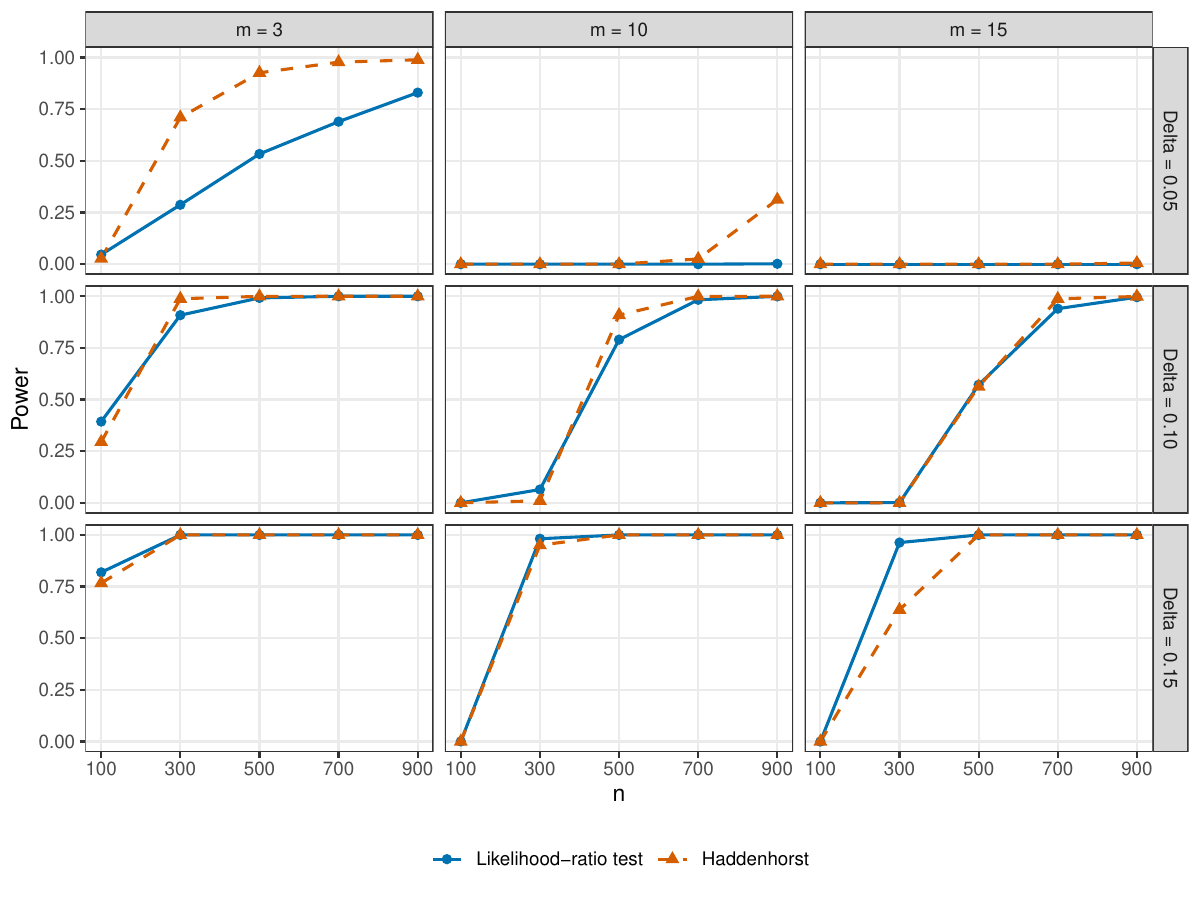}
    \caption{A plot of power. The blue solid line and the orange dashed line correspond to our likelihood-ratio test and the test of Haddenhorst et al. (2021), respectively. The rows correspond to $\Delta=0.05,0.10,0.15$ from top to bottom. The columns correspond to $m=3,10,15$ from left to right.}
    \label{fig:power}
\end{figure}

\section{Conclusion}\label{sec:conclusion}
In this paper, we proposed and studied a likelihood-ratio test of the null hypothesis that WST does not hold. We derived the uniform asymptotic size with an arbitrary number of items, and showed that the Type II error converges uniformly to zero under an appropriate signal-strength condition. Simulation studies demonstrated that the test has good finite-sample performance.

Several directions remain for future research. First, the signal-strength condition in Theorem \ref{thm:power} is only sufficient. It remains open whether this is also necessary. Second, developing tests for verifying other notions of stochastic transitivity, such as strong stochastic transitivity (SST), is an important direction. Such stronger notions can yield better sample-complexity guarantees for tasks such as estimating the pairwise win-probability matrix, identifying the best item, and ranking the items. For some of these tasks, WST alone may be insufficient \citep{shah2016stochastically, falahatgar2018limits}. Third, it is useful to extend our test to unbalanced or sparse comparison designs, where the number of comparisons varies across pairs and some pairs may be unobserved.

\section*{Acknowledgments}
%The authors thank the editor, associate editor, and anonymous referees for their valuable comments and suggestions, which significantly improved the paper. 
This work was supported by Grant-in-Aid for JSPS Fellows Grant Number JP25KJ1302, and by Nomura Foundation's Research Grant for Frontier of Financial and Capital Markets. The author reports there are no competing interests to declare.

\bibliographystyle{chicago}
\bibliography{ref}

\newpage
\appendix
\setcounter{page}{1}
\setcounter{section}{0}
\renewcommand{\theequation}{A.\arabic{equation}}
\setcounter{equation}{0}

\setcounter{table}{0}
\renewcommand{\thetable}{\thesection\arabic{table}}
\setcounter{figure}{0}
\renewcommand{\thefigure}{\thesection\arabic{figure}}
\numberwithin{equation}{section}

\begin{center}
{\Large Supplementary Material for} \\[7mm]
\textbf{\Large A Likelihood-Ratio Test for \\[3mm]
    Verifying Weak Stochastic Transitivity} \\[10mm]
\textsc{\large Masaki Toyoda$^*$} \\[5mm]
*\textit{\large Department of Economics, Hitotsubashi University}
\end{center}

\section{Proofs of Theorems}
\subsection{Proof of Theorem \ref{thm:m3reduction}}\label{proofthm:m3reduction}
\begin{proof}
    The desired equation immediately holds if $m=3$. Thus, we suppose that $m\geq4$.
    
    For any $m\geq4$, let
    \begin{align}
        \cT_m=\{(i,j,k):1\leq i<j<k\leq m\}.
    \end{align}
    For each $\tau=(\tau_1,\tau_2,\tau_3)\in\cT_m$, let $r_\tau=(p_{\tau_1\tau_2},p_{\tau_2\tau_3},1-p_{\tau_1\tau_3})$ and let $r_0=(1/2,1/2,1/2)$. Then, the null parameter space can be written as
    \begin{align}
        \cH_{0m}=\bigcup_{\tau\in\cT_m}\cH_{0m}^{(\tau)},~~\text{where}~~\cH_{0m}^{(\tau)}=\{p\in\Theta_m:r_\tau\in([1/2,1]^3\cup[0,1/2]^3)\setminus \{r_0\}\}.
    \end{align}
    Using this, it holds that
    \begin{align}
        T_{n,m}
        =\min_{\tau\in\cT_m}T_n^{(\tau)},~~\text{where}~~
        T_n^{(\tau)}=2\log\frac{\sup_{p\in\Theta_m}L_{n,m}(p)}{\sup_{p\in\cH_{0m}^{(\tau)}}L_{n,m}(p)}.
    \end{align}
    For any $p\in\cH_{0m}$, choose $\tau\in\cT_m$ such that $p\in\cH_{0m}^{(\tau)}$. For this $\tau$, let $p_\tau=(p_{\tau_1\tau_2},p_{\tau_2\tau_3},p_{\tau_1\tau_3})$. Then $p_{\tau}\in\cH_{03}$. Moreover, the distribution of $T_n^{(\tau)}$ under $p\in\cH_{0m}^{(\tau)}$ is the same as the distribution of $T_{n,3}$ under $p_\tau\in\cH_{03}$. This is because all terms corresponding to parameters other than $p_\tau$ are canceled from $T_n^{(\tau)}$. Therefore, for any $p\in\cH_{0m}$, we have
    \begin{align}
        \Pro_p(T_{n,m}>c)
        &=\Pro_p\left(\min_{\tau'\in\cT_m}T_n^{(\tau')}>c\right)
        \leq\Pro_p(T_n^{(\tau)}>c)\\
        &=\Pro_{p_\tau}(T_{n,3}>c)
        \leq\sup_{q\in\cH_{03}}\Pro_{q}(T_{n,3}>c).
    \end{align}
    Since this holds for any $p\in\cH_{0m}$, taking the supremum over $p\in\cH_{0m}$ yields
    \begin{align}\label{eq:supupper}
        \sup_{p\in\cH_{0m}}\Pro_p(T_{n,m}>c)\leq\sup_{q\in\cH_{03}}\Pro_q(T_{n,3}>c).
    \end{align}

    It remains to check $\sup_{p\in\cH_{0m}}\Pro_p(T_{n,m}>c)\geq\sup_{p\in\cH_{03}}\Pro_p(T_{n,3}>c)$. Fix any $q=(q_{12},q_{23},q_{13})\in\cH_{03}$ and let
    \begin{align}
        \tilde{q}_{ij}=
        \begin{cases}
            q_{ij}~&1\leq i<j\leq3\\
            1~&1\leq i<j\leq m~\text{and}~j\geq4.
        \end{cases}
    \end{align}
    Since $\tilde{q}=(\tilde{q}_{ij})_{1\leq i<j\leq m}\in\cH_{0m}^{((1,2,3))}\subset\cH_{0m}$, we have 
    \begin{align}
        \sup_{p\in\cH_{0m}}\Pro_p(T_{n,m}>c)
        &\geq\Pro_{\tilde{q}}(T_{n,m}>c)
        =\Pro_{\tilde{q}}\left(\min\left(T_n^{((1,2,3))},\min_{\tau\in \cT_m\setminus\{(1,2,3)\}}T_n^{(\tau)}\right)>c\right).
    \end{align}
    We bound $\min\left(T_n^{((1,2,3))},\min_{\tau\in \cT_m\setminus\{(1,2,3)\}}T_n^{(\tau)}\right)$. To this end, for any $\tau=(\tau_1,\tau_2,\tau_3)\in\cT_m$, let $\hat{r}^\tau=(\hat{r}_1^\tau,\hat{r}_2^\tau,\hat{r}_3^\tau)=(\hat{p}_{n,\tau_1\tau_2},\hat{p}_{n,\tau_2\tau_3},1-\hat{p}_{n,\tau_1\tau_3})$. Then
    \begin{align}
        T_n^{(\tau)}=\min(Q_n^{-(\tau)},Q_n^{+(\tau)}),
    \end{align}
    where
    \begin{align}
        &Q_n^{-(\tau)}=2n\sum_{i=1}^{3}\text{KL}\left(\hat{r}_{i}^\tau\;\middle\|\;\frac{1}{2}\right)\mathbf{1}\left\{\hat{r}_{i}^\tau<\frac{1}{2}\right\},\\
        &Q_n^{+(\tau)}=2n\sum_{i=1}^{3}\text{KL}\left(\hat{r}_{i}^\tau\;\middle\|\;\frac{1}{2}\right)\mathbf{1}\left\{\hat{r}_{i}^\tau>\frac{1}{2}\right\}~~\text{and}\\
        &\text{KL}(x\|1/2)=\begin{cases}
            x\log(2x)+(1-x)\log(2-2x)~&x\in(0,1)\\
            \log2~&x\in\{0,1\}.
        \end{cases}
    \end{align}
    First, we bound $T_n^{((1,2,3))}$. Among three components of $\hat{r}^{((1,2,3))}$, either at most one is strictly smaller than $1/2$, or at most one is strictly larger than $1/2$. Therefore, at least one of $Q_n^{-((1,2,3))}$ and $Q_n^{+((1,2,3))}$ contains at most one nonzero term. Thus,
    \begin{align}
        T_n^{((1,2,3))}
        \leq2n\sup_{x\in[0,1]}\text{KL}\left(x\;\middle\|\;\frac{1}{2}\right)
        =2n\log2.
    \end{align}
    Next, we bound $\min_{\tau\in \cT_m\setminus\{(1,2,3)\}}T_n^{(\tau)}$.
    \begin{itemize}
        \item If exactly two of $\tau_1,\tau_2,\tau_3$ belong to $\{1,2,3\}$, then, under $\tilde{q}$, $\hat{r}^\tau=(x,1,0)$ for some $x\in[0,1]$ a.s. If $x\leq1/2$, then $Q_n^{+(\tau)}=2n\log2\leq Q_n^{-(\tau)}$; if $x\geq1/2$, then $Q_n^{-(\tau)}=2n\log2\leq Q_n^{+(\tau)}$. Thus, $T_n^{(\tau)}=2n\log2$ a.s.
        \item If at most one of $\tau_1,\tau_2,\tau_3$ belongs to $\{1,2,3\}$, then, under $\tilde{q}$, $\hat{r}^\tau=(1,1,0)$ a.s. In this case, $Q_n^{-(\tau)}=2n\log2$ and $Q_n^{+(\tau)}=4n\log2$. Thus, $T_n^{(\tau)}=2n\log2$ a.s.
    \end{itemize} 
    Therefore, we have $\min\left(T_n^{((1,2,3))},\min_{\tau\in \cT_m\setminus\{(1,2,3)\}}T_n^{(\tau)}\right)=T_n^{((1,2,3))}$ a.s., which gives
    \begin{align}
        \Pro_{\tilde{q}}\left(\min\left(T_n^{((1,2,3))},\min_{\tau\in \cT_m\setminus\{(1,2,3)\}}T_n^{(\tau)}\right)>c\right)
        =\Pro_{\tilde{q}}(T_n^{((1,2,3))}>c).
    \end{align}
    Also, the distribution of $T_n^{((1,2,3))}$ under $\tilde{q}\in\cH_{0m}$ is the same as the distribution of $T_{n,3}$ under $q\in\cH_{03}$. This is because all terms corresponding to parameters other than $q$ are canceled from $T_n^{((1,2,3))}$. Consequently, for any $q\in\cH_{03}$, we have
    \begin{align}
         \sup_{p\in\cH_{0m}}\Pro_p(T_{n,m}>c)
         \geq\Pro_{\tilde{q}}(T_n^{((1,2,3))}>c)
         =\Pro_{q}(T_{n,3}>c).
    \end{align}
    Taking supremum over $q\in\cH_{03}$ yields
    \begin{align}\label{eq:suplower}
        \sup_{p\in\cH_{0m}}\Pro_p(T_{n,m}>c)\geq\sup_{q\in\cH_{03}}\Pro_q(T_{n,3}>c).
    \end{align}
    Combining \eqref{eq:supupper} and \eqref{eq:suplower} completes the proof.    
\end{proof}

\subsection{Proof of Theorem \ref{thm:type1}}\label{proofthm:type1}
\begin{proof}
    According to Theorem \ref{thm:m3reduction}, we have
    \begin{align}
        \sup_{p\in\cH_{0m_n}}\Pro_p(T_{n,m_n}>c)
        =\sup_{p\in\cH_{03}}\Pro_p(T_{n,3}>c).
    \end{align}
    Therefore, it is sufficient to show that the right-hand side converges to $(1/2)\Pro(\chi_1^2>c)+(1/4)\Pro(\chi_2^2>c)$.
    
    Let $r=(r_1,r_2,r_3)=(p_{12},p_{23},1-p_{13})$. For notational simplicity, we write $\Pro_r:=\Pro_{p(r)}$, where $p(r)=(r_1,r_2,1-r_3)$. Thus, $\Pro_r$ is simply $\Pro_p$ expressed in terms of the reparameterization $r$. Under this reparameterization, $\cH_{03}$ corresponds to $(C_+\cup C_-)\setminus\{r_0\}$, where $C_-=[0,1/2]^3$, $C_+=[1/2,1]^3$ and $r_0=(1/2,1/2,1/2)$. Therefore, we have
    \begin{align}
        &\sup_{p\in\cH_{03}}\Pro_p(T_{n,3}>c)
        =\max\left(\sup_{r\in C_+\setminus \{r_0\}}\Pro_r(T_{n,3}>c),~\sup_{r\in C_-\setminus \{r_0\}}\Pro_r(T_{n,3}>c)\right)\\
        &=\max\left(\sup_{r\in C_+}\Pro_r(T_{n,3}>c),~\sup_{r\in C_-}\Pro_r(T_{n,3}>c)\right)
        =\sup_{r\in C_+}\Pro_r(T_{n,3}>c),
    \end{align}
    where the second equality holds since $\Pro_r(T_{n,3}>c)$ is continuous in $r$ and $r_0$ is a limit point of both $C_+\setminus\{r_0\}$ and $C_-\setminus\{r_0\}$, and the last equality holds by symmetry under the transformation $r\mapsto1-r$, which interchanges $C_+$ and $C_-$ without changing the distribution of $T_{n,3}$.
    
    Moreover, for any $M>0$, define the local neighborhood of $(1/2,1/2,1/2)$ within $C_+$ and its complement by
    \begin{align}
        \cL_{n,M}=\left\{r\in C_+:\max_{i\in[3]}r_i\leq\frac{1}{2}+\frac{M}{2\sqrt{n}}\right\},~~
        \cN_{n,M}=\left\{r\in C_+:\max_{i\in[3]}r_i>\frac{1}{2}+\frac{M}{2\sqrt{n}}\right\}.
    \end{align}    
    Then, it holds that
    \begin{align}
        &\limsup_{n\to\infty}\sup_{r\in C_+}\Pro_r(T_{n,3}>c)\\
        &=\max\left(\limsup_{n\to\infty}\sup_{r\in \cN_{n,M}}\Pro_r(T_{n,3}>c),~\limsup_{n\to\infty}\sup_{r\in \cL_{n,M}}\Pro_r(T_{n,3}>c)\right)\\
        &\leq\frac{1}{2}\Pro(\chi_1^2>c)+\frac{1}{4}\Pro(\chi_2^2>c)+\frac{1}{M^2},
    \end{align}
    where the inequality holds by Lemmas \ref{theory:type1notlocal} and \ref{theory:type1local}. Since $M>0$ is arbitrary, taking $M\to\infty$ yields
    \begin{align}\label{eq:limsupTn}
        \limsup_{n\to\infty}\sup_{r\in C_+}\Pro_r(T_{n,3}>c)
        \leq\frac{1}{2}\Pro(\chi_1^2>c)+\frac{1}{4}\Pro(\chi_2^2>c).
    \end{align}
    On the other hand, by Lemma \ref{theory:lowerbound}, it also holds that
    \begin{align}\label{eq:liminfTn}
        \liminf_{n\to\infty}\sup_{r\in C_+}\Pro_r(T_{n,3}>c)
        \geq\frac{1}{2}\Pro(\chi_1^2>c)+\frac{1}{4}\Pro(\chi_2^2>c).
    \end{align}
    Combining \eqref{eq:limsupTn} and \eqref{eq:liminfTn} completes the proof.
\end{proof}

\subsection{Proof of Theorem \ref{thm:power}}\label{proofthm:power}
\begin{proof}
    For any $x,y\in[0,1]$, let $\text{KL}\left(x\;\middle\|\;y\right)$ denote the Bernoulli KL divergence with the usual convention at the boundary:
    \begin{align}\label{eq:KL}
        \text{KL}\left(x\;\middle\|\;y\right)=\begin{cases}
            0&x=y\\
            -\log(1-y)&x=0,~0<y<1\\
            -\log y&x=1,~0<y<1\\
            +\infty&x>0~\text{and}~y=0,~\text{or}~x<1~\text{and}~y=1\\
            x\log\cfrac{x}{y}+(1-x)\log\cfrac{1-x}{1-y}&\text{otherwise}.
        \end{cases}
    \end{align}
    Using this, we have the following bound.
    \begin{align}
        T_{n,m_n}
        &=2n\inf_{q\in\cH_{0m_n}}\sum_{i<j}\text{KL}\left(\hat{p}_{n,ij}\;\middle\|\;q_{ij}\right)
        \geq2n\inf_{q\in\cH_{0m_n}}\max_{i<j}\text{KL}\left(\hat{p}_{n,ij}\;\middle\|\;q_{ij}\right)\\
        &\geq4n\inf_{q\in\cH_{0m_n}}\max_{i<j}(\hat{p}_{n,ij}-q_{ij})^2
        =4n\left(\inf_{q\in\cH_{0m_n}}\max_{i<j}|\hat{p}_{n,ij}-q_{ij}|\right)^2,
    \end{align}
    where the second inequality uses Pinsker's inequality ($\text{KL}\left(x\;\middle\|\;y\right)\geq2(x-y)^2$).

    Moreover, the triangle inequality yields
    \begin{align}
        \inf_{q\in\cH_{0m_n}}\max_{i<j}|\hat{p}_{n,ij}-q_{ij}|
        \geq\inf_{q\in\cH_{0m_n}}\max_{i<j}|p_{ij}-q_{ij}|-\max_{i<j}|p_{ij}-\hat{p}_{n,ij}|.
    \end{align}
    
    To further bound this right-hand side, for any $n$ and any $p$, let $E_n=\{\max_{i<j}|p_{ij}-\hat{p}_{n,ij}|\leq\Delta_n/2\}$. If the event $E_n$ occurs and $p\in\cH_{1m_n}(\Delta_n)$ holds, then we have
    \begin{align}
        \inf_{q\in\cH_{0m_n}}\max_{i<j}|p_{ij}-q_{ij}|-\max_{i<j}|p_{ij}-\hat{p}_{n,ij}|\geq\Delta_n-\frac{\Delta_n}{2}=\frac{\Delta_n}{2}>0
    \end{align}
    and therefore, we also have
    \begin{align}
        4n\left(\inf_{q\in\cH_{0m_n}}\max_{i<j}|\hat{p}_{n,ij}-q_{ij}|\right)^2
        \geq4n\left(\frac{\Delta_n}{2}\right)^2=n\Delta_n^2.
    \end{align}
    Using this,
    \begin{align}
        &\sup_{p\in\cH_{1m_n}(\Delta_n)}\Pro_{p}(T_{n,m_n}\leq c)\\
        &\leq\sup_{p\in\cH_{1m_n}(\Delta_n)}\Pro_{p}(T_{n,m_n}\leq c,E_n)
        +\sup_{p\in\cH_{1m_n}(\Delta_n)}\Pro_{p}(E_n^c)\\
        &\leq\sup_{p\in\cH_{1m_n}(\Delta_n)}\Pro_{p}\left(n\Delta_n^2\leq c\right)+\sup_{p\in\cH_{1m_n}(\Delta_n)}\Pro_{p}(E_n^c)\\
        &=\mathbf{1}\left\{n\Delta_n^2\leq c\right\}+\sup_{p\in\cH_{1m_n}(\Delta_n)}\Pro_{p}(E_n^c).
    \end{align}
    Since Condition \ref{con:signal} implies $n\Delta_n^2\to\infty$, we have $\mathbf{1}\{n\Delta_n^2\leq c\}=o(1)$.

    On the other hand, it holds that
    \begin{align}
        \Pro_{p}(E_n^c)
        &\leq\sum_{i<j}\Pro_{p}\left(|\hat{p}_{n,ij}-p_{ij}|>\frac{\Delta_n}{2}\right)\\
        &\leq2\binom{m_n}{2}\exp\left\{-\frac{n\Delta_n^2}{2}\right\}
        \leq m_n^2\exp\left\{-\frac{n\Delta_n^2}{2}\right\},
    \end{align}
    where the first inequality is the union bound, the second inequality is Hoeffding's inequality, and the last inequality uses $2\binom{m_n}{2}=m_n(m_n-1)\leq m_n^2$. Since this right-hand side does not depend on $p$, taking supremum yields
    \begin{align}
        \sup_{p\in\cH_{1m_n}(\Delta_n)}\Pro_{p}(E_n^c)\leq m_n^2\exp\left\{-\frac{n\Delta_n^2}{2}\right\}
        =\exp\left\{2\log m_n-\frac{n\Delta_n^2}{2}\right\}
        =o(1),
    \end{align}
    where the last equality follows from Condition \ref{con:signal}. This completes the proof.
\end{proof}

\section{Lemmas and their Proofs}

\begin{lem}\label{theory:type1notlocal}
For any $M>0$ and $c>0$,
    \begin{align}
        \limsup_{n\to\infty}\sup_{r\in \cN_{n,M}}\Pro_r(T_{n,3}>c)
        \leq\frac{1}{2}\Pro(\chi_1^2>c)+\frac{1}{4}\Pro(\chi_2^2>c)+\frac{1}{M^2}.
    \end{align}
\end{lem}

\begin{proof}\label{prooftheory:type1notlocal}
    We begin with rewriting $T_{n,3}$. Let $\hat{r}_n=(\hat{r}_{1n},\hat{r}_{2n},\hat{r}_{3n})=(\hat{p}_{n,12},\hat{p}_{n,23},1-\hat{p}_{n,13})$. Also, for any $r=(r_1,r_2,r_3)\in[0,1]^3$, let
    \begin{align}
        \tilde{L}_n(r)=\prod_{i=1}^{3}r_i^{n\hat{r}_{in}}(1-r_i)^{n-n\hat{r}_{in}}.
    \end{align}
    Then, letting $r_0=(1/2,1/2,1/2)$, we have
    \begin{align}
        T_{n,3}&=2\min\left(\log\frac{\sup_{r\in[0,1]^3}\tilde{L}_n(r)}{\sup_{r\in C_+\setminus\{r_0\}}\tilde{L}_n(r)},~
        \log\frac{\sup_{r\in[0,1]^3}\tilde{L}_n(r)}{\sup_{r\in C_-\setminus\{r_0\}}\tilde{L}_n(r)}\right)\\
        &=2\min\left(\log\frac{\sup_{r\in[0,1]^3}\tilde{L}_n(r)}{\sup_{r\in C_+}\tilde{L}_n(r)},~
        \log\frac{\sup_{r\in[0,1]^3}\tilde{L}_n(r)}{\sup_{r\in C_-}\tilde{L}_n(r)}\right),
    \end{align}
    where the second equality holds since $\tilde{L}_n(r)$ is continuous in $r$ and $r_0$ is a limit point of both $C_+\setminus\{r_0\}$ and $C_-\setminus\{r_0\}$. Moreover,    
    \begin{align}
        &2\min\left(\log\frac{\sup_{r\in[0,1]^3}\tilde{L}_n(r)}{\sup_{r\in C_+}\tilde{L}_n(r)},~
        \log\frac{\sup_{r\in[0,1]^3}\tilde{L}_n(r)}{\sup_{r\in C_-}\tilde{L}_n(r)}\right)\\
        &=2n\min\left(\sum_{i=1}^{3}\text{KL}\left(\hat{r}_{in}\;\middle\|\;\frac{1}{2}\right)\mathbf{1}\left\{\hat{r}_{in}<\frac{1}{2}\right\},~\sum_{i=1}^{3}\text{KL}\left(\hat{r}_{in}\;\middle\|\;\frac{1}{2}\right)\mathbf{1}\left\{\hat{r}_{in}>\frac{1}{2}\right\}\right)\\
        &=\min(Q_n^-(\hat{r}_n),Q_n^+(\hat{r}_n)),
    \end{align}
    where
    \begin{align}
        &Q_n^-(\hat{r}_n)=2n\sum_{i=1}^{3}\text{KL}\left(\hat{r}_{in}\;\middle\|\;\frac{1}{2}\right)\mathbf{1}\left\{\hat{r}_{in}<\frac{1}{2}\right\},\\
        &Q_n^+(\hat{r}_n)=2n\sum_{i=1}^{3}\text{KL}\left(\hat{r}_{in}\;\middle\|\;\frac{1}{2}\right)\mathbf{1}\left\{\hat{r}_{in}>\frac{1}{2}\right\}.
    \end{align}
    Therefore, we have the following bound.
    \begin{align}
        &\Pro_r(T_{n,3}>c)
        =\Pro_r(Q_n^-(\hat{r}_n)>c,Q_n^+(\hat{r}_n)>c)
        \leq\Pro_r(Q_n^-(\hat{r}_n)>c)\\
        &=\Pro_r(Q_n^-(\hat{r}_n)>c,\hat{r}_{3n}<1/2)+\Pro_r(Q_n^-(\hat{r}_n)>c,\hat{r}_{3n}\geq1/2)\\
        &\leq\Pro_{r_3}\left(\hat{r}_{3n}<\frac{1}{2}\right)+\Pro_{r_1,r_2}\left(\sum_{i=1}^{2}2n\text{KL}\left(\hat{r}_{in}\;\middle\|\;\frac{1}{2}\right)\mathbf{1}\left\{\hat{r}_{in}<\frac{1}{2}\right\}>c\right).\label{eq:nonlocal}
    \end{align}

    If $r\in\cN_{n,M}$, there exists some $j\in[3]$ such that $r_j>1/2+M/(2\sqrt{n})$. By permutation symmetry, without loss of generality, we may assume that $r_3>1/2+M/(2\sqrt{n})$ throughout the proof.

    First, we bound the first term of \eqref{eq:nonlocal}. It holds that
    \begin{align}
        \Pro_{r_3}\left(\hat{r}_{3n}<\frac{1}{2}\right)
        \leq\Pro_{r_3}\left(|\hat{r}_{3n}-r_3|>r_3-\frac{1}{2}\right)
        \leq\frac{\text{Var}(\hat{r}_{3n})}{(r_3-1/2)^2}
        \leq\frac{1}{M^2},\label{eq:nonlocalfirst}
    \end{align}
    where the second inequality is Chebyshev's inequality, and the last inequality follows since $\text{Var}(\hat{r}_{3n})=r_3(1-r_3)/n\leq1/(4n)$ and $r_3-1/2>M/(2\sqrt{n})$.

    Next, we bound the second term of \eqref{eq:nonlocal}. Denote $D_n^-(x)=2n\text{KL}\left(x\;\middle\|\;1/2\right)\mathbf{1}\{x<1/2\}$. Since $\hat{r}_{1n}$ and $\hat{r}_{2n}$ are independent, and $r\in\cN_{n,M}\subset C_+$ implies $r_1,r_2\geq1/2$, applying Lemma \ref{theory:dominance} successively yields
    \begin{align}
        &\Pro_{r_1,r_2}\left(\sum_{i=1}^{2}2n\text{KL}\left(\hat{r}_{in}\;\middle\|\;\frac{1}{2}\right)\mathbf{1}\left\{\hat{r}_{in}<\frac{1}{2}\right\}>c\right)
        =\Pro_{r_1,r_2}(D_n^-(\hat{r}_{1n})+D_n^-(\hat{r}_{2n})>c)\\
        &\leq\Pro_{1/2,r_2}(D_n^-(\hat{r}_{1n})+D_n^-(\hat{r}_{2n})>c)
        \leq\Pro_{1/2,1/2}(D_n^-(\hat{r}_{1n})+D_n^-(\hat{r}_{2n})>c).\label{eq:nonlocalsecond}
    \end{align}
    Combining \eqref{eq:nonlocal}, \eqref{eq:nonlocalfirst} and \eqref{eq:nonlocalsecond}, we have
    \begin{align}
        \Pro_r(T_{n,3}>c)\leq\frac{1}{M^2}+\Pro_{1/2,1/2}(D_n^-(\hat{r}_{1n})+D_n^-(\hat{r}_{2n})>c).
    \end{align}
    Since the right-hand side does not depend on $r$, it holds that
    \begin{align}
        \sup_{r\in\cN_{n,M}}\Pro_r(T_{n,3}>c)\leq\frac{1}{M^2}+\Pro_{1/2,1/2}(D_n^-(\hat{r}_{1n})+D_n^-(\hat{r}_{2n})>c).
    \end{align}
    It remains to evaluate $\limsup_{n\to\infty}\Pro_{1/2,1/2}(D_n^-(\hat{r}_{1n})+D_n^-(\hat{r}_{2n})>c)$. For each $i=1,2$, under $r_i=1/2$, the central limit theorem implies $2\sqrt{n}(\hat{r}_{in}-1/2)\xrightarrow{d}N(0,1)$. Furthermore, since $D_n^-(\hat{r}_{in})=(2\sqrt{n}(\hat{r}_{in}-1/2))^2\mathbf{1}\{\hat{r}_{in}<1/2\}+o_p(1)$, the continuous mapping theorem and Slutsky's lemma imply that $D_n^-(\hat{r}_{in})\xrightarrow{d}Z_i^2\mathbf{1}\{Z_i<0\}$, where $Z_i\sim N(0,1)$. Moreover, since $D_n^-(\hat{r}_{1n})$ and $D_n^-(\hat{r}_{2n})$ are independent, we have $D_n^-(\hat{r}_{1n})+D_n^-(\hat{r}_{2n})\xrightarrow{d}Z_1^2\mathbf{1}\{Z_1<0\}+Z_2^2\mathbf{1}\{Z_2<0\}$, where $Z_1,Z_2\sim N(0,1)$ independent. Since $c>0$ is a continuity point of the limiting distribution, we have
    \begin{align}
        \limsup_{n\to\infty}\sup_{r\in\cN_{n,M}}\Pro_r(T_{n,3}>c)&\leq\frac{1}{M^2}+\Pro(Z_1^2\mathbf{1}\{Z_1<0\}+Z_2^2\mathbf{1}\{Z_2<0\}>c)\\
        &=\frac{1}{M^2}+\frac{1}{2}\Pro(\chi_1^2>c)+\frac{1}{4}\Pro(\chi_2^2>c).
    \end{align}
    This completes the proof.
\end{proof}

\begin{lem}\label{theory:dominance}
    Suppose that $r_i\geq1/2$. For any $t\in\bbR$, we have
    \begin{align}
        \Pro_{r_i}(D_n^-(\hat{r}_{in})>t)\leq\Pro_{1/2}(D_n^-(\hat{r}_{in})>t).
    \end{align}
\end{lem}
\begin{proof}
    Recall that $D_n^-(x)=2n\text{KL}\left(x\;\middle\|\;1/2\right)\mathbf{1}\{x<1/2\}$, which is continuous and non-increasing in $x\in[0,1]$. Hence, for any $t\in\bbR$, there exists some $c_t\in\bbR$ such that $\{D_n^-(x)>t\}=\{x<c_t\}$. Using this, we have
    \begin{align}
        \Pro_{r_i}(D_n^-(\hat{r}_{in})>t)
        =\Pro_{r_i}(\hat{r}_{in}<c_t)
        \leq\Pro_{1/2}(\hat{r}_{in}<c_t)
        =\Pro_{1/2}(D_n^-(\hat{r}_{in})>t),
    \end{align}
    where the inequality uses the stochastic dominance of binomial distributions. This completes the proof.
\end{proof}

\begin{lem}\label{theory:type1local}
For any $M>0$,
    \begin{align}
        \limsup_{n\to\infty}\sup_{r\in \cL_{n,M}}\Pro_r(T_{n,3}>c)
        \leq\frac{1}{2}\Pro(\chi_1^2>c)+\frac{1}{4}\Pro(\chi_2^2>c).
    \end{align}
\end{lem}
\begin{proof}
   For any $w=(w_1,w_2,w_3)\in\bbR^3$, let $g(w)=\min(\sum_{i=1}^{3}w_i^2\mathbf{1}\{w_i<0\},\sum_{i=1}^{3}w_i^2\mathbf{1}\{w_i>0\})$. Also, let $W_n=(W_{1n},W_{2n},W_{3n})$ where $W_{in}=2\sqrt{n}(\hat{r}_{in}-1/2)$. Then, for any $0<\rho<c/2$, we have
    \begin{align}
        \Pro_r(T_{n,3}>c)
        &\leq\Pro_r(T_{n,3}>c,|T_{n,3}-g(W_n)|\leq\rho)+\Pro_r(|T_{n,3}-g(W_n)|>\rho)\\
        &\leq\Pro_r(g(W_n)>c-\rho)+\Pro_r(|T_{n,3}-g(W_n)|>\rho).
    \end{align}
    Moreover, let
    \begin{align}
        \psi_\rho(x)=
        \begin{cases}
            0&x\leq c-2\rho\\
            \cfrac{x-(c-2\rho)}{\rho}&c-2\rho<x<c-\rho\\
            1&c-\rho\leq x.
        \end{cases}
    \end{align}
    Then we have $\mathbf{1}\{x>c-\rho\}\leq\psi_\rho(x)\leq\mathbf{1}\{x>c-2\rho\}$. Letting $\E_r$ be expectation under $\Pro_r$, we have
    \begin{align}
        \Pro_r(g(W_n)>c-\rho)
        =\E_r[\mathbf{1}\{g(W_n)>c-\rho\}]
        \leq\E_r[\psi_\rho(g(W_n))].
    \end{align}
    Letting $b_n=(b_{1n},b_{2n},b_{3n})$ where $b_{in}=2\sqrt{n}(r_i-1/2)$, and $Z\sim N(0,I_3)$, we can bound this as
    \begin{align}
        \E_r[\psi_\rho(g(W_n))]
        &\leq|\E_r[\psi_\rho(g(W_n))]-\E[\psi_\rho(g(b_n+Z))]|+\E[\psi_\rho(g(b_n+Z))]\\
        &\leq|\E_r[\psi_\rho(g(W_n))]-\E[\psi_\rho(g(b_n+Z))]|+\E[\mathbf{1}\{g(b_n+Z)>c-2\rho\}]\\
        &=|\E_r[\psi_\rho(g(W_n))]-\E[\psi_\rho(g(b_n+Z))]|+\Pro(g(b_n+Z)>c-2\rho).
    \end{align}
    Consequently, taking supremum yields
    \begin{align}
        &\sup_{r\in\cL_{n,M}}\Pro_r(T_{n,3}>c)\\
        &\leq\sup_{r\in\cL_{n,M}}|\E_r[\psi_\rho(g(W_n))]-\E[\psi_\rho(g(b_n+Z))]|\\
        &\quad\quad+\sup_{r\in\cL_{n,M}}\Pro(g(b_n+Z)>c-2\rho)
        +\sup_{r\in\cL_{n,M}}\Pro_r(|T_{n,3}-g(W_n)|>\rho).
    \end{align}
    Applying Lemma \ref{theory:psiconverge} to the first term, Lemma \ref{theory:gaussian} to the second term, and Lemma \ref{theory:tnconverge} to the last term yields
    \begin{align}
        \limsup_{n\to\infty}\sup_{r\in\cL_{n,M}}\Pro_r(T_{n,3}>c)\leq\frac{1}{2}\Pro(\chi_1^2>c-2\rho)+\frac{1}{4}\Pro(\chi_2^2>c-2\rho).
    \end{align}
    Since $0<\rho<c/2$ is arbitrary, and both $\chi_1^2$ and $\chi_2^2$ have continuous distributions, letting $\rho\to0$ completes the proof.
\end{proof}

\begin{lem}\label{theory:tnconverge}
    For any $M>0$ and $t>0$, it holds that
    \begin{align}
        \limsup_{n\to\infty}\sup_{r\in\cL_{n,M}}\Pro_r(|T_{n,3}-g(W_n)|>t)=0,
    \end{align}
\end{lem}
\begin{proof}
    For any $R>M$, we have
    \begin{align}
        &\Pro_r(|T_{n,3}-g(W_n)|>t)\\
        &\leq\Pro_r(\|W_n\|_\infty>R)+\Pro_r(|T_{n,3}-g(W_n)|>t,\|W_n\|_\infty\leq R).\label{eq:tn}
    \end{align}
    First, we bound the first term of the right-hand side of \eqref{eq:tn}. Recall that $W_{in}=2\sqrt{n}(\hat{r}_{in}-1/2)$. Therefore,
    \begin{align}
        \Pro_r(\|W_n\|_\infty>R)
        &\leq\sum_{i=1}^{3}\Pro_r\left(2\sqrt{n}\left|\hat{r}_{in}-\frac{1}{2}\right|>R\right)\\
        &\leq\sum_{i=1}^{3}\Pro_r\left(2\sqrt{n}\left(|\hat{r}_{in}-r_i|+\left|r_i-\frac{1}{2}\right|\right)>R\right).
    \end{align}
    If $r\in\cL_{n,M}$, we have $|r_i-1/2|\leq M/(2\sqrt{n})$ for each $i\in[3]$. Thus,
    \begin{align}
        \sup_{r\in\cL_{n,M}}\Pro_r(\|W_n\|_\infty>R)
        &\leq\sum_{i=1}^{3}\sup_{r\in\cL_{n,M}}\Pro_r\left(2\sqrt{n}|\hat{r}_{in}-r_i|+M>R\right)\\
        &\leq\sum_{i=1}^{3}\sup_{r\in\cL_{n,M}}\frac{4r_i(1-r_i)}{(R-M)^2}
        \leq\frac{3}{(R-M)^2},\label{eq:tn1}
    \end{align}
    where the second inequality is Chebyshev's inequality.

    Next, we bound the second term of the right-hand side of \eqref{eq:tn}. For any $n\in\bbN$ and $-\sqrt{n}\leq x\leq\sqrt{n}$, define
    \begin{align}
        H_n(x)=2n\text{KL}\left(\frac{1}{2}+\frac{x}{2\sqrt{n}}\;\middle\|\;\frac{1}{2}\right).
    \end{align}
    Then, we have
    \begin{align}
        Q_n^-(\hat{r}_n)=\sum_{i=1}^{3}2n\text{KL}\left(\hat{r}_{in}\;\middle\|\;\frac{1}{2}\right)\mathbf{1}\left\{\hat{r}_{in}<\frac{1}{2}\right\}=\sum_{i=1}^{3}H_n(W_{in})\mathbf{1}\{W_{in}<0\},\\
        Q_n^+(\hat{r}_n)=\sum_{i=1}^{3}2n\text{KL}\left(\hat{r}_{in}\;\middle\|\;\frac{1}{2}\right)\mathbf{1}\left\{\hat{r}_{in}>\frac{1}{2}\right\}=\sum_{i=1}^{3}H_n(W_{in})\mathbf{1}\{W_{in}>0\}.
    \end{align}
    Thus, we have $T_{n,3}=\min(\sum_{i=1}^{3}H_n(W_{in})\mathbf{1}\{W_{in}<0\},\sum_{i=1}^{3}H_n(W_{in})\mathbf{1}\{W_{in}>0\})$. Applying Lemma \ref{theory:minmax} to $T_{n,3}$ and $g(W_n)$ yields
    \begin{align}
        &|T_{n,3}-g(W_n)|\\
        &\leq\max\left(\left|\sum_{i=1}^{3}(H_n(W_{in})-W_{in}^2)\mathbf{1}\{W_{in}<0\}\right|,\left|\sum_{i=1}^{3}(H_n(W_{in})-W_{in}^2)\mathbf{1}\{W_{in}>0\}\right|\right)\\
        &\leq\sum_{i=1}^{3}|H_n(W_{in})-W_{in}^2|~~\text{a.s.}
    \end{align}
    For sufficiently large $n$ such that $n\geq2R^2$, on the event $\{\|W_n\|_\infty\leq R\}$, we have
    \begin{align}
        \sum_{i=1}^{3}|H_n(W_{in})-W_{in}^2|
        \leq3\sup_{-R\leq w\leq R}|H_n(w)-w^2|
        \leq\frac{12R^4}{n},
    \end{align}
    where the last inequality holds by Lemma \ref{theory:hn}.
    
    Consequently, we have
    \begin{align}
        \Pro_r(|T_{n,3}-g(W_n)|>t,\|W_n\|_\infty\leq R)
        \leq\mathbf{1}\left\{\frac{12R^4}{n}>t\right\}.\label{eq:tn2}
    \end{align}
    Combining \eqref{eq:tn}, \eqref{eq:tn1} and \eqref{eq:tn2} yields
    \begin{align}
        \sup_{r\in\cL_{n,M}}\Pro_r(|T_{n,3}-g(W_n)|>t)
        \leq\frac{3}{(R-M)^2}+\mathbf{1}\left\{\frac{12R^4}{n}>t\right\}.
    \end{align}
    Thus,
    \begin{align}
        \limsup_{n\to\infty}\sup_{r\in\cL_{n,M}}\Pro_r(|T_{n,3}-g(W_n)|>t)
        \leq\frac{3}{(R-M)^2}.
    \end{align}
    Since this inequality holds for any $R>M$, letting $R\to\infty$ completes the proof.
\end{proof}

\begin{lem}\label{theory:minmax}
    For any $a,b,c,d\in\bbR$, it holds that
    \begin{align}
        |\min(a,b)-\min(c,d)|\leq\max(|a-c|,|b-d|).
    \end{align}
\end{lem}
\begin{proof}
    Let $\delta=\max(|a-c|,|b-d|)$. Then $c-\delta\leq a\leq c+\delta$ and $d-\delta\leq b\leq d+\delta$. Using this, we have
    \begin{align}
        \min(c,d)-\delta
        =\min(c-\delta,d-\delta)
        \leq\min(a,b)
        \leq\min(c+\delta,d+\delta)
        =\min(c,d)+\delta.
    \end{align}
    Therefore,
    \begin{align}
        |\min(a,b)-\min(c,d)|\leq\delta=\max(|a-c|,|b-d|).
    \end{align}
    This completes the proof.
\end{proof}

\begin{lem}\label{theory:hn}
    For any $R>0$, if $n\geq2R^2$, it holds that
    \begin{align}
        \sup_{-R\leq w\leq R}|H_n(w)-w^2|\leq\frac{4R^4}{n}.
    \end{align}
\end{lem}
\begin{proof}
    Let $f(x)=(1+x)\log(1+x)+(1-x)\log(1-x)$. Then, we can write $H_n(w)=nf(w/\sqrt{n})$.

    For $|x|<1$, the Taylor expansion of $f(x)$ around $x=0$ gives
    \begin{align}
        f(x)=2\sum_{m=1}^{\infty}\frac{x^{2m}}{(2m-1)2m}.
    \end{align}
    Since $|w/\sqrt{n}|\leq R/\sqrt{n}\leq1/\sqrt{2}<1$, we have
    \begin{align}
        H_n(w)&=nf\left(\frac{w}{\sqrt{n}}\right)
        =n\sum_{m=1}^{\infty}\frac{2(w/\sqrt{n})^{2m}}{(2m-1)2m}
        =w^2+\sum_{m=2}^{\infty}\frac{2w^{2m}}{(2m-1)2mn^{m-1}}.
    \end{align}
    Clearly, $H_n(w)-w^2\geq0$. Thus,
    \begin{align}
        \sup_{-R\leq w\leq R}|H_n(w)-w^2|
        \leq\sum_{m=2}^{\infty}\frac{2R^{2m}}{n^{m-1}}
        =2n\sum_{m=2}^{\infty}\left(\frac{R^{2}}{n}\right)^m
        =\frac{2R^4/n}{1-R^2/n}.
    \end{align}
    Since $n\geq2R^2$, it holds that $1-R^2/n\geq1/2$. Thus,
    \begin{align}
        \frac{2R^4/n}{1-R^2/n}
        \leq\frac{2R^4/n}{1/2}
        =\frac{4R^4}{n}.
    \end{align}
    This completes the proof.
\end{proof}

\begin{lem}\label{theory:psiconverge}
    For any $M>0$ and $0<\rho<c/2$, it holds that
    \begin{align}
        \limsup_{n\to\infty}\sup_{r\in\cL_{n,M}}|\E_r[\psi_\rho(g(W_n))]-\E[\psi_\rho(g(b_n+Z))]|=0.
    \end{align}
\end{lem}
\begin{proof}  
    Let $S_n=(S_{1n},S_{2n},S_{3n})=W_n-b_n=2\sqrt{n}(\hat{r}_n-r)$. Then $S_{1n},S_{2n},S_{3n}$ are independent. We also take $Z_1,Z_2,Z_3$ to be independent of $S_n$. 
    
    Throughout the proof, subscripts on $\E$ indicate the variables over which expectation is taken, with $r_i$ referring to $S_{in}$. Let
    \begin{align}
        E_0&=\E_{r_1,r_2,r_3}[\psi_\rho(g(b_n+(S_{1n},S_{2n},S_{3n})))],\\
        E_1&=\E_{Z_1,r_2,r_3}[\psi_\rho(g(b_n+(Z_{1},S_{2n},S_{3n})))],\\
        E_2&=\E_{Z_1,Z_2,r_3}[\psi_\rho(g(b_n+(Z_{1},Z_{2},S_{3n})))],~\text{and}\\
        E_3&=\E_{Z_1,Z_2,Z_3}[\psi_\rho(g(b_n+(Z_{1},Z_{2},Z_{3})))].
    \end{align}
    Then,
    \begin{align}
        &|\E_r[\psi_\rho(g(W_n))]-\E[\psi_\rho(g(b_n+Z))]|\\
        &=|E_0-E_3|\\
        &\leq|E_0-E_1|+|E_1-E_2|+|E_2-E_3|.
    \end{align}
    We bound each of these terms. For sufficiently large $n$ such that $n\geq 4M^2$, we have
    \begin{align}
        &|E_0-E_1|\\
        &=\left|\E_{r_2,r_3}\left[\E_{r_1}[\psi_\rho(g(b_n+(S_{1n},S_{2n},S_{3n})))]
        -\E_{Z_1}[\psi_\rho(g(b_n+(Z_1,S_{2n},S_{3n})))]\right]\right|\\
        &\leq\E_{r_2,r_3}\left[\left|\E_{r_1}[\psi_\rho(g(b_n+(S_{1n},S_{2n},S_{3n})))]
        -\E_{Z_1}[\psi_\rho(g(b_n+(Z_1,S_{2n},S_{3n})))]\right|\right]\\
        &\leq\E_{r_2,r_3}\left[\frac{C}{\sqrt{n}}\right]
        =\frac{C}{\sqrt{n}},
    \end{align}
    where the second inequality follows by applying Lemma \ref{theory:psilip} conditionally on $S_{2n}$ and $S_{3n}$. Applying Lemma \ref{theory:psilip} successively to $|E_1-E_2|$ and $|E_2-E_3|$, we have
    \begin{align}
        |E_1-E_2|\leq\frac{C}{\sqrt{n}},~~|E_2-E_3|\leq\frac{C}{\sqrt{n}},
    \end{align}
    and thus,
    \begin{align}
        |\E_r[\psi_\rho(g(W_n))]-\E[\psi_\rho(g(b_n+Z))]|\leq\frac{3C}{\sqrt{n}}.
    \end{align}
    Since $C$ does not depend on $r$ and $n$, it holds that
    \begin{align}
        \limsup_{n\to\infty}\sup_{r\in\cL_{n,M}}|\E_r[\psi_\rho(g(W_n))]-\E[\psi_\rho(g(b_n+Z))]|=0.
    \end{align}
    This completes the proof.
\end{proof}

\begin{lem}\label{theory:psilip}
    For any $M>0$, $0<\rho<c/2$ and $n\geq4M^2$, let $r\in\cL_{n,M}$. Then, the following bounds hold uniformly over $(s_1,s_2,s_3)\in\bbR^3$:
    \begin{align}
        &|\E_{r_1}[\psi_\rho(g(b_n+(S_{1n},s_2,s_3)))]
        -\E_{Z_1}[\psi_\rho(g(b_n+(Z_1,s_2,s_3)))]|
        \leq\frac{C}{\sqrt{n}},\\
        &|\E_{r_2}[\psi_\rho(g(b_n+(s_1,S_{2n},s_3)))]
        -\E_{Z_2}[\psi_\rho(g(b_n+(s_1,Z_2,s_3)))]|
        \leq\frac{C}{\sqrt{n}},~~\text{and}\\
        &|\E_{r_3}[\psi_\rho(g(b_n+(s_1,s_2,S_{3n})))]
        -\E_{Z_3}[\psi_\rho(g(b_n+(s_1,s_2,Z_3)))]|
        \leq\frac{C}{\sqrt{n}},
    \end{align}
    where $C$ can be chosen such that
    \begin{align}
        C=\frac{2\sqrt{c-\rho}}{\rho}\left(6+M^2\right).
    \end{align}
\end{lem}
\begin{proof}
    For any $x=(x_1,x_2,x_3)\in\bbR^3$, let
    \begin{align}
        m(x)=\min\left(\sqrt{\sum_{i=1}^{3}x_i^2\mathbf{1}\{x_i<0\}},~\sqrt{\sum_{i=1}^{3}x_i^2\mathbf{1}\{x_i>0\}}\right).
    \end{align}
    Then, for any $x,y\in\bbR^3$, we have $|\psi_\rho(g(x))-\psi_\rho(g(y))|=|\psi_\rho((m(x))^2)-\psi_\rho((m(y))^2)|$.

    The function $s\mapsto\psi_\rho(s^2)$ is continuous on $[0,\infty)$ and differentiable except at finitely many points, $\sqrt{c-2\rho}$ and $\sqrt{c-\rho}$. At every point of differentiability, we have
    \begin{align}
        \left|\frac{d}{ds}(\psi_\rho(s^2))\right|\leq\frac{2\sqrt{c-\rho}}{\rho}.
    \end{align}
    Hence, $\psi_\rho(s^2)$ is $2\sqrt{c-\rho}/\rho$-Lipschitz on $[0,\infty)$. Thus, 
    \begin{align}
        |\psi_\rho((m(x))^2)-\psi_\rho((m(y))^2)|\leq\frac{2\sqrt{c-\rho}}{\rho}|m(x)-m(y)|
        \leq\frac{2\sqrt{c-\rho}}{\rho}\|x-y\|_2,
    \end{align}
    where the last inequality holds by Lemma \ref{theory:mlip}. Therefore, a function $h:\bbR\to\bbR$ defined as
    \begin{align}
        h(t)=\frac{\rho}{2\sqrt{c-\rho}}\psi_{\rho}(g(b_n+(t,s_2,s_3)))
    \end{align}
    satisfies $h\in\text{Lip}(1)$. Using this, we have
    \begin{align}
        &|\E_{r_1}[\psi_\rho(g(b_n+(S_{1n},s_2,s_3)))]
        -\E_{Z_1}[\psi_\rho(g(b_n+(Z_1,s_2,s_3))))]|\\
        &=\frac{2\sqrt{c-\rho}}{\rho}|\E_{r_1}[h(S_{1n})]-\E_{Z_1}[h(Z_1)]|\\
        &\leq\frac{2\sqrt{c-\rho}}{\rho}\left|\E_{r_1}[h(S_{1n})]-\E_{r_1}\left[h\left(\frac{S_{1n}}{2\sqrt{r_1(1-r_1)}}\right)\right]\right|\\
        &\quad+\frac{2\sqrt{c-\rho}}{\rho}\left|\E_{r_1}\left[h\left(\frac{S_{1n}}{2\sqrt{r_1(1-r_1)}}\right)\right]-\E_{Z_1}[h(Z_1)]\right|.\label{eq:h}
    \end{align}
    To bound the right-hand side, we introduce a proposition:
    \begin{prop}\label{theory:wassersteinclt}[\citesuppl{dobler2015new}, Corollary 2.10.]
    Let $X_1,X_2,\dots$ be i.i.d. with $\E[X_1]=0$ and $\text{Var}(X_1)=1$. Let $Z\sim N(0,1)$. Then,
    \begin{align}
        \sup_{h\in\text{Lip}(1)}\left|\E\left[h\left(\frac{1}{\sqrt{n}}\sum_{i=1}^{n}X_i\right)\right]-\E[h(Z)]\right|\leq\frac{3\E[|X_1|^3]}{\sqrt{n}},
    \end{align}
    where $\text{Lip}(1)=\left\{f:\bbR\to\bbR:|f(x)-f(y)|\leq |x-y|~\forall x,y\in\bbR\right\}$.
\end{prop}

    By definition, $S_{1n}=W_{1n}-b_{1n}=(1/\sqrt{n})\sum_{t=1}^{n}2(X_{t,12}-r_1)$, where $\E[2(X_{t,12}-r_1)]=0$ and $\text{Var}(2(X_{t,12}-r_1))=4r_1(1-r_1)$. Thus, Proposition \ref{theory:wassersteinclt} implies
    \begin{align}
        \left|\E_{r_1}\left[h\left(\frac{S_{1n}}{2\sqrt{r_1(1-r_1)}}\right)\right]-\E_{Z_1}\left[h\left(Z_1\right)\right]\right|
        &\leq\frac{3}{\sqrt{n}}\E_{r_1}\left[\left|\frac{X_{t,12}-r_1}{\sqrt{r_1(1-r_1)}}\right|^3\right]\\
        &=\frac{3(r_1^2+(1-r_1)^2)}{\sqrt{nr_1(1-r_1)}}.\label{eq:h1}
    \end{align}
    Since $n\geq4M^2$ and $r\in\cL_{n,M}$, we have $1/2\leq r_1\leq3/4$, which implies $r_1^2+(1-r_1)^2\leq5/8$ and $r_1(1-r_1)\geq3/16$. Therefore,
    \begin{align}
        \frac{3(r_1^2+(1-r_1)^2)}{\sqrt{nr_1(1-r_1)}}
        \leq\frac{15}{2\sqrt{3n}}
        \leq\frac{6}{\sqrt{n}}.
    \end{align}

    On the other hand, we have
    \begin{align}
        &\left|\E_{r_1}[h(S_{1n})]-\E_{r_1}\left[h\left(\frac{S_{1n}}{2\sqrt{r_1(1-r_1)}}\right)\right]\right|
        \leq\E_{r_1}\left[\left|h(S_{1n})-h\left(\frac{S_{1n}}{2\sqrt{r_1(1-r_1)}}\right)\right|\right]\\
        &\leq\left|1-\frac{1}{2\sqrt{r_1(1-r_1)}}\right|\E_{r_1}\left[\left|S_{1n}\right|\right]
        \leq\left(\frac{1}{2\sqrt{r_1(1-r_1)}}-1\right)\sqrt{\E_{r_1}\left[S_{1n}^2\right]}\\
        &=\left(\frac{1}{2\sqrt{r_1(1-r_1)}}-1\right)2\sqrt{r_1(1-r_1)}
        =1-2\sqrt{r_1(1-r_1)}\\
        &\leq1-2\sqrt{\left(\frac{1}{2}+\frac{M}{2\sqrt{n}}\right)\left(\frac{1}{2}-\frac{M}{2\sqrt{n}}\right)}
        =1-\sqrt{1-\frac{M^2}{n}}
        \leq\frac{M^2}{n},\label{eq:h2}
    \end{align}
    where the first inequality is Jensen, the second inequality holds since $h\in\text{Lip}(1)$ and $r_1\geq1/2$, third inequality is Cauchy-Schwarz, fourth inequality holds since $r_1\leq1/2+M/(2\sqrt{n})$, and the last inequality uses $1-\sqrt{1-x}\leq x$ for $0\leq x\leq1$.

    Combining \eqref{eq:h}, \eqref{eq:h1} and \eqref{eq:h2} yields
    \begin{align}
        &|\E_{r_1}[\psi_\rho(g(b_n+(S_{1n},s_2,s_3)))]
        -\E_{Z_1}[\psi_\rho(g(b_n+(Z_1,s_2,s_3)))]|\\
        &\leq\frac{2\sqrt{c-\rho}}{\rho}\left(\frac{6}{\sqrt{n}}+\frac{M^2}{n}\right)
        \leq\frac{2\sqrt{c-\rho}}{\rho}\frac{6+M^2}{\sqrt{n}}
        =\frac{C}{\sqrt{n}}.
    \end{align}
    This bound does not depend on $s_2$ and $s_3$, and hence is uniform over them. The corresponding bounds for the second and third coordinates follow by permutation symmetry. This completes the proof.
\end{proof}

\begin{lem}\label{theory:mlip}
    For any $x,y\in\bbR^3$, it holds that
    \begin{align}
        |m(x)-m(y)|\leq\|x-y\|_2.
    \end{align}
\end{lem}
\begin{proof}
    For any $s\in\bbR$, let $s^-=\max(-s,0)$ and $s^+=\max(s,0)$. Also, let $x^-=(x_1^-,x_2^-,x_3^-)$, $x^+=(x_1^+,x_2^+,x_3^+)$, $y^-=(y_1^-,y_2^-,y_3^-)$ and $y^+=(y_1^+,y_2^+,y_3^+)$. Then, $m(x)=\min(\|x^-\|_2,\|x^+\|_2)$ and $m(y)=\min(\|y^-\|_2,\|y^+\|_2)$. By Lemma \ref{theory:minmax} and the reverse triangle inequality, we have
    \begin{align}
        |m(x)-m(y)|
        &\leq\max(|\|x^-\|_2-\|y^-\|_2|, |\|x^+\|_2-\|y^+\|_2|)\\
        &\leq\max(\|x^--y^-\|_2,\|x^+-y^+\|_2).
    \end{align}
    Since the maps $s\mapsto s^-$ and $s\mapsto s^+$ are $1$-Lipschitz, we have $\|x^--y^-\|_2\leq\|x-y\|_2$ and $\|x^+-y^+\|_2\leq \|x-y\|_2$. Therefore,
    \begin{align}
        |m(x)-m(y)|\leq\|x-y\|_2.
    \end{align}
    This completes the proof.
\end{proof}

\begin{lem}\label{theory:gaussian}
    For any $b=(b_1,b_2,b_3)\in[0,\infty)^3$ and any $t>0$, it holds that
    \begin{align}
        \Pro(g(b+Z)>t)\leq\frac{1}{2}\Pro(\chi_1^2>t)+\frac{1}{4}\Pro(\chi_2^2>t).
    \end{align}
\end{lem}
\begin{proof}
    By Lemma \ref{theory:schur}, we have
    \begin{align}
        &\Pro(g(b+Z)>t)\\
        &\leq\Pro(g((b_1+b_2+b_3+Z_1,Z_2,Z_3))>t)\\
        &\leq\Pro(Z_2^2\mathbf{1}\{Z_2<0\}+Z_3^2\mathbf{1}\{Z_3<0\}>t,b_1+b_2+b_3+Z_1\geq0)\\
        &\quad+\Pro(Z_2^2\mathbf{1}\{Z_2>0\}+Z_3^2\mathbf{1}\{Z_3>0\}>t,b_1+b_2+b_3+Z_1<0).
    \end{align}
    By independence of $Z_1$ and $(Z_2,Z_3)$, and since $(Z_2,Z_3)\overset{d}{=}(-Z_2,-Z_3)$, it holds that
    \begin{align}
        &\Pro(Z_2^2\mathbf{1}\{Z_2<0\}+Z_3^2\mathbf{1}\{Z_3<0\}>t,b_1+b_2+b_3+Z_1\geq0)\\
        &\quad+\Pro(Z_2^2\mathbf{1}\{Z_2>0\}+Z_3^2\mathbf{1}\{Z_3>0\}>t,b_1+b_2+b_3+Z_1<0)\\
        &=\Pro(b_1+b_2+b_3+Z_1\geq0)\Pro(Z_2^2\mathbf{1}\{Z_2<0\}+Z_3^2\mathbf{1}\{Z_3<0\}>t)\\
        &\quad+\Pro(b_1+b_2+b_3+Z_1<0)\Pro(Z_2^2\mathbf{1}\{Z_2>0\}+Z_3^2\mathbf{1}\{Z_3>0\}>t)\\
        &=\Pro(Z_2^2\mathbf{1}\{Z_2>0\}+Z_3^2\mathbf{1}\{Z_3>0\}>t)\\
        &=\frac{1}{2}\Pro(\chi_1^2>t)+\frac{1}{4}\Pro(\chi_2^2>t).
    \end{align}
    This completes the proof.
\end{proof}

\begin{lem}\label{theory:schur}
    For any $b_1,b_2,b_3\geq0$ and any $t>0$, it holds that
    \begin{align}
        \Pro(g((b_1+Z_1,b_2+Z_2,b_3+Z_3))>t)\leq\Pro(g((b_1+b_2+b_3+Z_1,Z_2,Z_3))>t).
    \end{align}
\end{lem}
\begin{proof}
    Let
    \begin{align}
        U=\frac{(b_1+Z_1)+(b_2+Z_2)}{2},~
        V=\frac{(b_1+Z_1)-(b_2+Z_2)}{2}.
    \end{align}
    Then, $b_1+Z_1=U+V$ and $b_2+Z_2=U-V$, and $U,V,Z_3$ are mutually independent. Thus,
    \begin{align}
        &g((b_1+Z_1,b_2+Z_2,b_3+Z_3))\\
        &=\min(((U+V)^-)^2+((U-V)^-)^2+((b_3+Z_3)^-)^2,\\
        &\quad\quad\quad((U+V)^+)^2+((U-V)^+)^2+((b_3+Z_3)^+)^2),
    \end{align}
    where $x^-=\max(-x,0)$ and $x^+=\max(x,0)$. Moreover, both components in the minimum are non-decreasing in $|V|$. Indeed,
    \begin{enumerate}
        \item If $U\geq0$ and $0\leq|V|\leq U$, then $((U+|V|)^-)^2+((U-|V|)^-)^2=0$ and $((U+|V|)^+)^2+((U-|V|)^+)^2=2U^2+2|V|^2$.
        \item If $U\geq0$ and $0\leq U<|V|$, then $((U+|V|)^-)^2+((U-|V|)^-)^2=(U-|V|)^2$ and $((U+|V|)^+)^2+((U-|V|)^+)^2=(U+|V|)^2$.
        \item If $U<0$ and $0\leq|V|\leq-U$, then $((U+|V|)^-)^2+((U-|V|)^-)^2=2U^2+2|V|^2$ and $((U+|V|)^+)^2+((U-|V|)^+)^2=0$.
        \item If $U<0$ and $0\leq-U<|V|$, then $((U+|V|)^-)^2+((U-|V|)^-)^2=(U-|V|)^2$ and $((U+|V|)^+)^2+((U-|V|)^+)^2=(U+|V|)^2$.
    \end{enumerate}
    Therefore, the minimum is also non-decreasing in $|V|$. Hence, for any $t$, there exists a nonnegative measurable function $f_t$ such that
    \begin{align}
        \{g((b_1+Z_1,b_2+Z_2,b_3+Z_3))>t\}
        =\{|V|>f_t(U,b_3+Z_3)\}~~\text{a.s.}
    \end{align}
    Thus, we have
    \begin{align}
        &\Pro(g((b_1+Z_1,b_2+Z_2,b_3+Z_3))>t)\\&=\E[\Pro(g((b_1+Z_1,b_2+Z_2,b_3+Z_3))>t\mid U,b_3+Z_3)]\\
        &=\E[\Pro(|V|>f_t(U,b_3+Z_3)\mid U,b_3+Z_3)].
    \end{align}
    To bound this, consider
    \begin{align}
        V'=\frac{(b_1+b_2+Z_1)-Z_2}{2}.
    \end{align}
    Then $U,V',Z_3$ are mutually independent, and $V\sim N((b_1-b_2)/2,1/2)$, $V'\sim N((b_1+b_2)/2,1/2)$. Since $|\E[V]|=|b_1-b_2|/2\leq(b_1+b_2)/2=\E[V']$, applying Lemma \ref{theory:normalincreasing} gives
    \begin{align}
        &\E[\Pro(|V|>f_t(U,b_3+Z_3)\mid U,b_3+Z_3)]\\
        &\leq\E\left[\Pro\left(\left|V'\right|>f_t(U,b_3+Z_3)\mid U,b_3+Z_3\right)\right]\\
        &=\E[\Pro(g((b_1+b_2+Z_1,Z_2,b_3+Z_3))>t\mid U,b_3+Z_3)]\\
        &=\Pro(g((b_1+b_2+Z_1,Z_2,b_3+Z_3))>t).
    \end{align}
    By permutation symmetry, the same argument can be applied to any pair of coordinates. Applying it to the first and third coordinates gives 
    \begin{align}
        \Pro(g((b_1+b_2+Z_1,Z_2,b_3+Z_3))>t)
        \leq\Pro(g((b_1+b_2+b_3+Z_1,Z_2,Z_3))>t).
    \end{align}
    This completes the proof.
\end{proof}

\begin{lem}\label{theory:normalincreasing}
    For any $\mu_1,\mu_2\in\bbR$, any $\sigma>0$, let $X_1\sim N(\mu_1,\sigma^2)$ and $X_2\sim N(\mu_2,\sigma^2)$. If $|\mu_1|\leq|\mu_2|$, then for any $a\geq0$, it holds that
    \begin{align}
        \Pro\left(\left|X_1\right|>a\right)
        \leq\Pro\left(\left|X_2\right|>a\right).
    \end{align}
\end{lem}
\begin{proof}
    Let $H_a(\mu)=\Pro(|X_\mu|>a)$ where $X_\mu\sim N(\mu,\sigma^2)$. Since $H_a(\mu)$ is an even function, it is sufficient to check that $H_a(\mu)$ is non-decreasing when $\mu\geq0$. It holds that
    \begin{align}
        H_a(\mu)=1-\Phi\left(\frac{a-\mu}{\sigma}\right)+\Phi\left(\frac{-a-\mu}{\sigma}\right),
    \end{align}
    where $\Phi$ is the distribution function of $N(0,1)$. Thus,
    \begin{align}
        H_a'(\mu)=\frac{1}{\sigma}\left(\phi\left(\frac{a-\mu}{\sigma}\right)-\phi\left(\frac{a+\mu}{\sigma}\right)\right)\geq0,
    \end{align}
    where $\phi$ is the density function of $N(0,1)$. This completes the proof.
\end{proof}

\begin{lem}\label{theory:lowerbound}
For any $c>0$,
    \begin{align}
        \liminf_{n\to\infty}\sup_{r\in C_+}\Pro_r(T_{n,3}>c)
        \geq\frac{1}{2}\Pro(\chi_1^2>c)+\frac{1}{4}\Pro(\chi_2^2>c).
    \end{align}
\end{lem}
\begin{proof}
    Let $r^*=(1/2,1/2,1)\in C_+$. Then,
    \begin{align}
        \sup_{r\in C_+}\Pro_r(T_{n,3}>c)\geq\Pro_{r^*}(T_{n,3}>c)
        =\Pro_{r^*}(\min(Q_n^-(\hat{r}_n),Q_n^+(\hat{r}_n))>c),
    \end{align}
    where
    \begin{align}
        Q_n^-(\hat{r}_n)&=\sum_{i=1}^{3}D_n^-(\hat{r}_{in}),~~
        Q_n^+(\hat{r}_n)=\sum_{i=1}^{3}D_n^+(\hat{r}_{in}),\\
        D_n^-(x)&=
        2n\text{KL}\left(x\;\middle\|\;\frac{1}{2}\right)\mathbf{1}\left\{x<\frac{1}{2}\right\},~~
        D_n^+(x)=2n\text{KL}\left(x\;\middle\|\;\frac{1}{2}\right)\mathbf{1}\left\{x>\frac{1}{2}\right\}.
    \end{align}
    Under $r^*$, we have $\hat{r}_{3n}=1$ a.s. Therefore, $Q_n^-(\hat{r}_n)=D_n^-(\hat{r}_{1n})+D_n^-(\hat{r}_{2n})$ and $Q_n^+(\hat{r}_n)\geq2n\text{KL}\left(1\;\middle\|\;\frac{1}{2}\right)=2n\log2$ a.s. Thus, for all sufficiently large $n$ such that $2n\log2>c$, we have
    \begin{align}
        \Pro_{r^*}(\min(Q_n^-(\hat{r}_n),Q_n^+(\hat{r}_n))>c)
        =\Pro_{1/2,1/2}(D_n^-(\hat{r}_{1n})+D_n^-(\hat{r}_{2n})>c).
    \end{align}
    As shown at the end of the proof of Lemma \ref{theory:type1notlocal}, $\Pro_{1/2,1/2}(D_n^-(\hat{r}_{1n})+D_n^-(\hat{r}_{2n})>c)\to(1/2)\Pro(\chi_1^2>c)+(1/4)\Pro(\chi_2^2>c)$. Consequently,
    \begin{align}
        \liminf_{n\to\infty}\sup_{r\in C_+}\Pro_r(T_{n,3}>c)
        &\geq\liminf_{n\to\infty}\Pro_{1/2,1/2}(D_n^-(\hat{r}_{1n})+D_n^-(\hat{r}_{2n})>c)\\
        &=\frac{1}{2}\Pro(\chi_1^2>c)+\frac{1}{4}\Pro(\chi_2^2>c).
    \end{align}
    This completes the proof.
\end{proof}

\section{Implementation Details for the Haddenhorst et al. (2021) Procedure}\label{sec:haddenhorst}
In Section \ref{sec:simulations}, we compare our likelihood-ratio test with Algorithm 2 of \cite{haddenhorst2021testing}. Their procedure controls the probabilities of incorrectly declaring non-WST when WST holds, and incorrectly declaring WST when WST does not hold. Under our formulation, in which the null hypothesis is that WST does not hold, only the latter error is a Type I error. We therefore calibrate their procedure to control this error at our nominal level $\alpha$.

The procedure of \cite{haddenhorst2021testing} is based on a directed graph $G$ with vertex set $[m]$, and the set of edges is determined as follows. For each $i<j$ and each $t\in[n]$, it sequentially checks whether
\begin{align}\label{eq:sprt}
    \text{(a)}:~\hat{p}_{t,ij}>\frac{1}{2}+C_t^\text{SPRT}~~\text{or}~~
    \text{(b)}:~\hat{p}_{t,ij}<\frac{1}{2}-C_t^\text{SPRT},
\end{align}
where
\begin{align}
    C_t^\text{SPRT}=\frac{1}{2t}\left\lceil\frac{\log{(1-\gamma)/\gamma}}{\log{(1/2+h)/(1/2-h)}}\right\rceil,~~
    \gamma=\frac{\alpha}{\binom{m}{2}-\lfloor (m+1)/3\rfloor}.
\end{align}
If (a) holds before (b), then the edge $i\to j$ is added to $G$; if (b) holds before (a), then the edge $j\to i$ is added to $G$; if neither (a) nor (b) holds, then the edge remains undetermined.

In our implementation, we repeatedly generate a random permutation of all $\binom{m}{2}$ pairs and process one additional comparison for each pair in that order. After each edge is added, the procedure examines the current graph, and rejects the null hypothesis as soon as every completion of the undetermined edges of $G$ is acyclic. It stops without rejection if $G$ contains a directed cycle. If neither stopping condition is satisfied before all $n\binom{m}{2}$ comparisons have been processed, the null hypothesis is not rejected. Theorem 8.4 of \cite{haddenhorst2021testing} implies that, if $|p_{ij}-1/2|>h$ for every pair $i<j$, this procedure has Type I error at most $\alpha$. In our simulation, we set $h=0.04$ and $\alpha=0.05$.

\bibliographystylesuppl{chicago}
\bibliographysuppl{ref}

\end{document}